\documentclass[11pt,A4paper]{article}

\usepackage[margin=1in]{geometry}
\usepackage{amsfonts}
\usepackage{graphicx,color}
\usepackage{xspace}
\usepackage{enumerate}
\usepackage{amssymb,amsmath}
\usepackage{amsthm}
\usepackage{subfig}
\usepackage{setspace}
\usepackage{color}

\usepackage{comment}

\usepackage{breakcites}

\usepackage[T1]{fontenc}
\usepackage[utf8]{inputenc}
\usepackage{authblk}
\usepackage[breaklinks]{hyperref}

\newtheorem{theorem}{Theorem}

\newtheorem{corollary}[theorem]{Corollary}
\newtheorem{proposition}[theorem]{Proposition}
\newtheorem{lemma}[theorem]{Lemma}

\newtheorem{example}[theorem]{Example}

\newtheorem{definition}[theorem]{Definition}

\newcommand{\comp}{\ensuremath{\textrm{cr}}\xspace}

\begin{document}

\author[1]{Spyros Angelopoulos}

\author[2]{Thomas Lidbetter}

\author[3]{Konstantinos Panagiotou}

\date{}

\affil[1]{CNRS and LIP6-Sorbonne University}

\affil[2]{Rutgers Business School}

\affil[3]{University of Munich}

\title{Search Games with Predictions}

\maketitle

\begin{abstract}
We introduce the study of {\em search games} between a mobile Searcher and an immobile Hider in a new setting in which the Searcher has some potentially erroneous information, i.e., a {\em prediction} on the Hider's position.  The objective is to establish tight tradeoffs between the {\em consistency} of a search strategy (i.e., its worst case expected payoff assuming the prediction is correct) and its {\em robustness} (i.e., the worst case expected payoff with no assumptions on the quality of the prediction). Our study is the first to address the full power of mixed (randomized) strategies; previous work focused only on deterministic strategies, or relied on stochastic assumptions that do not guarantee worst-case robustness in adversarial situations. We give Pareto-optimal strategies for three fundamental problems, namely searching in discrete locations, searching with stochastic overlook, and searching in the infinite line. As part of our contribution, we provide a novel framework for proving optimal tradeoffs in search games which is applicable, more broadly, to any two-person zero-sum games in learning-augmented settings.


\end{abstract}

\noindent

\section{Introduction}
\label{sec:intro}

Searching for a hidden target is a common task in everyday life, and an important computational setting with numerous real-world applications. Problems related to search arise in such diverse areas as drilling for oil in multiple sites, the deployment of search-and-rescue operations, or robot navigation in unknown terrains. We are interested, specifically, in the formulation in which a mobile {\em Searcher} must locate an immobile {\em Hider} that lies in some unknown position within the {\em search space}, i.e., the environment in which the search takes place. There is, also, some underlying concept of the quality of search, which captures, informally, how quickly the Searcher can locate the Hider, or equivalently, for how long the Hider can evade being found. 

Within the field of Operations Research, search problems have been studied under the mathematical formulation of a zero-sum two person game between the Searcher and the Hider. Here, the two players define first their respective set of {\em pure strategies}, then follow a probabilistic choice among these pure strategies, which gives rise to a {\em mixed strategy}. The quality of the (pure or mixed) strategy is reflected in an appropriately defined {\em payoff}, which the Searcher seeks to minimize, whereas the Hider seeks to maximize.  The objective, in this formulation, is to identify the {\em value} of the search game, namely the expected payoff of Searcher and Hider strategies that are at equilibrium or, equivalently, the expected payoff of the corresponding best-response strategies. This game-theoretic framework has been applied successfully to a multitude of search problems, and {\em search games} has developed into a prolific field; we refer the reader to books such as~\cite{isaacs65},~\cite{searchgames},~\cite{alpern2013search},~\cite{Gal80}.

Search games have also been studied from a computational standpoint, often with an explicit distinction between pure (i.e., deterministic) and mixed (i.e., randomized) strategies. Here, the emphasis is both on the efficiency and the computability of the strategies, in the sense that approximate strategies are sought if the underlying problem is computationally hard. Many search problems have been studied from this TCS perspective, see e.g.,~\cite{koutsoupias:fixed},~\cite{yates:plane},~\cite{alex:robots},\cite{jaillet:online}~\cite{ray:2randomized},~\cite{hybrid},~\cite{ultimate},~\cite{kupavskii2018lower},~\cite{DBLP:journals/algorithmica/BonatoGMP23},~\cite{kalyanasundaram1993competitive},~\cite{demaine:turn},~\cite{spiral} for some representative examples of search-related algorithms.

\subsection{Search games with predictions}

Previous work on search games has focused predominantly on games in which no assumptions are made by one player about the strategy the other player will choose. In practice, however, such information is often both available and useful, notwithstanding its inherent inaccuracy. For example, in a search-and-rescue operation, the rescue team may use witness information about the whereabouts of the missing person, which often may  lead to a much faster tracking of the person. To formulate such settings, let $H$ denote the {\em prediction} available to the Searcher, i.e., a subset of the Hider's strategies, which will be defined according to the search game at hand. For games over {\em pure} strategies,~\cite{DBLP:conf/innovations/000121} studied a model in which the quality of a search strategy $S$ is determined by two values $(C(S),R(S))$.  Here, $C(S)$ is the {\em consistency} of $S$, namely the worst-case payoff of $S$ assuming that the prediction $H$ is {\em accurate}, in that it represents faithfully and error-freely the corresponding information on the position of the Hider. On the other hand, $R(S)$ is the {\em robustness} of $S$, namely the worst-case payoff of $S$ with no assumptions on the accuracy of the prediction. These two objectives are, clearly, in a tradeoff relation, and the goal  is to identify the best-possible such tradeoff. Using terminology from multi-objective optimization, we seek the {\em Pareto frontier}~\cite{boyd2004convex} of the game, which describes the best-possible consistency, under the condition that the robustness does not exceed some given value $r$. An alternative definition seeks the best robustness that can be achieved, under the condition that the consistency does not exceed a given value $c$. The above definitions are borrowed from {\em learning-augmented} online computation, a field that has witnessed remarkable growth recently, and in which the online algorithm is enhanced with some machine-learned prediction; see, e.g.,~\cite{DBLP:conf/icml/LykourisV18},~\cite{NIPS2018_8174}, the survey~\cite{mitzenmacher_vassilvitskii_2021} and the online collection~\cite{predictionslist} of several works in the past few years. 

This framework carries over to the game-theoretical formulation, i.e., when the players have mixed (randomized) strategies. However, as we will discuss, randomization poses new challenges, which are the main focus of this work.

\subsection{Contribution}
\label{subsec:contribution}
We study three important games from search theory in learning-augmented settings with fully randomized player strategies. The first problem (Section~\ref{sec:box-finite}) is that of a Searcher who wishes to find a Hider among a set of discrete locations (e.g., boxes), each of which has a designated search cost (e.g., the cost for opening each box), which we interpret as the time taken to search that location (or {\em search time}). The Searcher opens boxes one by one, until she finds the Hider, and the objective is to minimize the total time spent searching. The aim is to find min-max mixed (randomized) strategies for the Searcher (and max-min strategies for the Hider.) For the standard, no predictions setting, a solution of the special case with unit costs can be found in~\cite{ruckle1991discrete}, while for general costs, solutions can be found in~\cite{Condon:2009:ADA:1497290.1497300}. The work~\cite{Condon:2009:ADA:1497290.1497300} also highlights the relation of the game to database query optimization, in which queries are evaluated via ``pipelined'' filtering, a problem with many applications from quality control to machine learning.
Other solutions can also be found in~\cite{Lidbetter:multiple} and \cite{AL:expanding}. In our setting, the (natural) prediction is a subset of the boxes in which the Hider is predicted to hide\footnote{Note that our model allows for $H$ to define a {\em set} of potential predictions instead of a single prediction. We refer to Section~\ref{sec:preliminaries} for the formal definition of consistency in this context, which is in line with the multiple-predictions framework of~\cite{anand2022online}.}.

Our second problem (Section~\ref{sec:box-infinite}) extends the box search game to a setting of {\em imperfect} detection. Here, there is a known detection probability $q$, in the sense that each time a box containing the Hider is searched, the Hider is found independently with probability $q$.
As for the standard box search game with no overlook probabilities, the aim is to find min-max and max-min strategies for the Searcher and Hider, respectively. Variants of this game (without predictions) have been studied in recent works~\cite{clarkson2023classical},~\cite{clarkson2024computing} and a closed-form solution was given in~\cite{bui2024optimal}. The problem can help formulate applications related to stochastic product testing and query evaluation, when the testing oracle is only partially reliable~\cite{clarkson2023classical}. As in the standard box search, in our setting the prediction is expressed as a subset of boxes in which the Hider is expected to be.



Our last problem is perhaps the most fundamental search game in an unbounded environment. Namely, in Section~\ref{sec:linear.search} we consider the problem of {\em linear search} (informally known in TCS as the {\em cow-path} problem), in which the Hider hides at some point in an infinite line, while the Searcher is initially located to some point designated as the {\em origin} and must locate the Hider by alternatively exploring the left and the right halflines (relative to the origin). In its standard version, this problem was first solved by Gal~\cite{Gal80}, and independently in~\cite{ray:2randomized}. Linear search is one of the classic problems in competitive analysis, and has close connections to other resource allocation problems under uncertainty, such as the design of interruptible algorithms~\cite{steins, DBLP:journals/scheduling/Angelopoulos22}. In our work, we describe the Pareto frontier in the setting in which there is some natural directional prediction about the position of the Hider, i.e., whether it hides to the left or to the right of the origin. This problem was previously studied in the more restricted setting of pure strategies in~\cite{DBLP:conf/innovations/000121}.

For all the above problems, we provide Pareto-optimal strategies that achieve optimal consistency-robustness tradeoffs. To our knowledge, this is the first work to analyze such tradeoffs for randomized algorithms (mixed strategies) not only within the field of search games, but more generally for games with imperfect predictions. We emphasize that randomization poses new challenges in the analysis of the underlying games, notably in regards to the application of the {\em minimax} theorem for obtaining lower bounds on the performance of the search strategies. For problems such as box search, we achieve this by showing tight lower bounds on the robustness of any search strategy with a given consistency $c$. However, this approach is not immediately applicable to other problems such as searching on the line, since it is not clear how to characterize the class of mixed search strategies of bounded consistency (or robustness). To bypass this complication, in Section~\ref{sec:general} we use ideas from multi-objective optimization and show that it suffices to find the value of a new game, in which the Searcher seeks to minimize a linear combination of consistency and robustness. This provides a new methodology that is applicable not only to search games, but to any two-person zero-sum game with predictions, and does not rely on any explicit characterization of best-response strategies. 

One might expect that for box search games, strategies that are Pareto optimal can be obtained by randomizing between the optimal strategy if the prediction is assumed to be true, and the optimal strategy if the prediction is ignored (an approach often used in learning-augmented online algorithms). Indeed, for box search with perfect detection, this approach is fruitful, but it cannot work for the imperfect detection variant, in which every strategy of finite robustness must be infinite. The reason is that if the prediction is assumed to be correct, then the optimal strategy would be to repeatedly search the boxes indicated by the prediction, and never search any other boxes. But if this strategy is played with any positive probability, then with this probability the other boxes will never be searched, leading to infinite robustness. Therefore, it is necessary to construct strategies that delicately balance the need to give preference to boxes suggested by the prediction and the need to ensure that the other boxes are not neglected. A similar situation arises in the linear search game, where the optimal strategy exhibits a carefully chosen ``bias'' towards the predicted half-line.


\subsection{Other related work}
\label{subsec:related}

The work~\cite{DBLP:conf/innovations/000121} introduced the setting of Pareto-based analysis of deterministic search algorithms, and showed upper and lower bounds on the consistency/robustness tradeoffs for pure strategies, in the context of the linear search problem. Subsequent work~\cite{DBLP:conf/mfcs/000123} improved some of these bounds, and extended others to the $m$-ray search problems, assuming again pure strategies. \cite{eberle2022robustification} showed how to robustify graph exploration algorithms that leverage a prediction on the minimum spanning tree of the explored graph. \cite{DBLP:conf/innovations/BanerjeeC0L23} and~\cite{depavia2024learning} studied graph search algorithms in a setting in which vertex of the graph can provide an imperfect estimation of the distance of the said vertex to the hider, whereas~\cite{cabello2024searching} studied learning-augmented searching in Euclidean Spaces. 

There are several works related to exploring a known or unknown search space in the {\em advice complexity} model. Here, the explorer is enhanced with an additional, and error-free information called {\em advice}, which may originate from a very powerful oracle. The objective is to quantify the number of advice bits required to achieve a desired performance, typically measured by the {\em competitive ratio}, see., e.g.,~\cite{dobrev2012online},~\cite{dobrev2012online},~\cite{gorain2018deterministic},~\cite{komm2015treasure},~\cite{pelc2021advice}. These works focus on the performance of deterministic strategies, and they assume that the prediction is error-free.  

Some previous work has studied settings in which the Searcher relies on advice in the form of noisy queries~\cite{boczkowski2018searching},~\cite{DBLP:journals/corr/abs-2303-13378},~\cite{DBLP:journals/corr/abs-2111-09093}. For instance, when searching a binary tree, an oracle may advice the Searcher to proceed to the left or to the right subtree of the currently visited node. In this setting, the response to each query is correct with some fixed probability that is known to the Searcher. In our setting, in contrast, we do not rely on any probabilistic assumptions with regards to the quality of the advice. Last, we note that finding Pareto-optimal algorithms is the main focus of several works in learning-augmented online computation, 
e.g.,~\cite{sun2021pareto,wei2020optimal,li2021robustness,lee2021online,lee2022pareto,christianson2023optimal,DBLP:conf/innovations/0001DJKR20,DBLP:conf/eenergy/Lee0HL24, bamas2020primal}. 

\section{Preliminaries}
\label{sec:preliminaries}

We consider the abstraction of a zero-sum search game $G$ between Player~1 (the minimizer Searcher) and Player~2 (the maximizer Hider). In the standard, no-predictions setting, we will refer to $G$ as the {\em standard} game. Let $X$ be the set of Player~1's pure strategies and let $Y$ be the set of Player~2's pure strategies.  We generally assume that $X$ and $Y$ are both finite, except for the linear search problem in Section~\ref{sec:linear.search}, where they are both infinite. Let $\Delta(X)$ and $\Delta(Y)$ be the players'  mixed strategy sets (i.e., the set of probability distributions over $X$ and $Y$, respectively). Let $u(x,y)$ denote the expected {\em payoff} when Player~1 plays $x\in \Delta(X)$ and Player~2 plays ${y \in \Delta(Y)}$. We also extend the notation to allow the argument of the payoff $u$ to take values in $X$ and $Y$, i.e., in pure strategies.  Note that the definition of the payoff is particular to the game in question.
Using a convention from the theory of search games, we will use small letters to refer to mixed strategies (e.g., $s$) and capital letters to refer to pure strategies (e.g., $S$).

A {\em prediction} is given by a proper subset $H \subset Y$ (so that $H \neq \emptyset, Y$). For a given $x \in \Delta(X)$ and a given prediction $H$, we define the {\em consistency} $C(x)$ and the {\em robustness} $R(x)$ by
\[
C(x)  = \sup_{y \in H} u(x,y) \ \textrm{ and }
R(x) = \sup_{y \in Y} u(x,y).
\]
In other words, the consistency is the worst case expected payoff of $x$, given that the prediction is correct (that is, Player~2 chooses a strategy in $H$), whereas the robustness is the worst-case expected payoff regardless of whether or not the prediction is correct. Note that the consistency has a dependency on $H$, which we suppress in the notation, since $H$ is considered to be fixed. 

For a given prediction $H$, we would like to minimize both the consistency and robustness, and so we seek to understand the set $S$ of pairs $(C(x),R(x))$, and in particular to characterize the set $P$ of Player~1 strategies that are {\em Pareto optimal}. More precisely, a strategy $x \in \Delta(X)$ is Pareto optimal if there does not exist any $x'$ such that $C(x') \le C(x)$ and $R(x') \le R(x)$, with at least one of these inequalities being strict.

Observe that if the game $G$ has a value $V = \inf_{x \in \Delta(X)} \sup_{y \in Y} u(x,y) = \sup_{y \in \Delta(Y)} \inf_{x \in X} u(x,y)$, then no strategy $x\in \Delta(X)$ can have robustness less than $V$. If the game obtained by restricting Player 2's strategy set to $H$ has a value $V'$, then no strategy $x \in \Delta(X)$ can have consistency less than $V'$. Furthermore, no strategy $x \in \Delta(X)$ on the Pareto frontier can have consistency greater than $V$, otherwise its robustness would also be greater than $V$, so it would be Pareto dominated by any optimal Player 1 strategy in $G$ (or some $\varepsilon$-optimal strategy, for small enough $\varepsilon$).

\section{Box Search}
\label{sec:box}

In this section, we study a fundamental search game in which a Hider (Player 2) hides in one of a set of boxes and a Searcher (Player 1) looks in the boxes one by one until she finds the Hider. We consider two variations of the game. In the first variation, which we use as a warm-up for the more complex setting, the Searcher is guaranteed to find the Hider when opening the box where it is located, whereas in the second variation, there is a given probability with which the Hider is detected if the Searcher opens the box in which it is located.

\subsection{Box search with perfect detection} \label{sec:box-finite}

In this game, the Hider's strategy set $Y$ is the set $\{1,2,\ldots,n\}$ of boxes and the Searcher's strategy set $X$ is the set of permutations of $Y$.  
We denote by $t_j$ the search {\em time} required to search box $j \in Y$. For a given permutation $\sigma \in X$ and a given box $j \in Y$, the payoff $u(\sigma,j)$ (also called the {\em search time} of $\sigma$ against $j$) is defined as the sum of the search times $t_i$ of the boxes $i$ opened up to and including box $j$. 

We will first state the known solution to the standard game, giving an optimal searcher strategy, which we denote by $s_Y$. Here, for any $Z\subseteq Y$  we denote by $s_Z$ the strategy that chooses the first box in $Z$ with probability proportional to its search time, then opens the remaining boxes in $Z$ in a uniformly random order.

\begin{theorem}[\cite{Lidbetter:multiple}]\label{thm:box}
An optimal (min-max) strategy for the Searcher in box search with perfect detection is given by  the search  $s=s_Y \in \Delta(X)$.  This results in an expected search time (equal to the value $V(Y)$ of the game) of
\begin{align}
 V(Y)=\frac{t^2(Y) + t(Y)^2}{2t(Y)}, \label{eq:box-value}
\end{align}
where $t^2(Y) = \sum_{j \in Y} t_j^2$ and $t(Y) = \sum_{j \in Y} t_j$.
\end{theorem}
\cite{Condon:2009:ADA:1497290.1497300} independently gave an alternative optimal randomized strategy. Theorem~\ref{thm:box} was proved in~\cite{Lidbetter:multiple} by using a lower bound provided by the following useful lemma. 
\begin{lemma} \label{lem:hiding}
Suppose the target is hidden in box $j$ with probability $h^*_Y (j) \equiv t_j/t(Y)$ proportional to the search time of box $j$. Then any search strategy against $h^*_Y$ has expected search time $V(Y)$.
\end{lemma}

We now move to our learning-augmented setting, in which there is a given  prediction $H \subset [n]$, and let us denote by $H^c$ the complement of $H$. We will show that, in order to obtain a Pareto-optimal strategy, the Searcher should mix between two strategies: one that first performs an optimal search $s_H$ of $H$ followed by an optimal search $s_{H^c}$ of $H^c$, and one in which the order of optimal searches in $H$ and $H^c$ is reversed (c.f. Definition~\ref{def:s_box} and Lemma~\ref{lem:P-opt}). In this strategy, the higher the probability of searching $H$ first, the higher the consistency and the lower the robustness. We will prove that the consistency and robustness of the family of strategies we will introduce are related by a linear equation. We will also obtain a  tight lower bound by proving that any mixed Searcher strategy satisfies a corresponding linear inequality (Lemma~\ref{lem:box-lb}); this is accomplished by exploiting knowledge of optimal Hider strategies in the game.

We will use the following elementary lemma.
\begin{lemma}\label{lem:V-identity}
For any $H \subseteq Y$,
\begin{align} 
t(Y) V(Y) &= t(H) V(H) + t(H^c) V(H^c) + t(H)t(H^c), \label{eq:V}
\end{align}
where $H^c$ denotes the complement of $H$.
\end{lemma}

\begin{proof}
The right-hand side of~(\ref{eq:V}) is, by definition
\begin{align*}
t(H) V(H) + t(H^c) V(H^c) + t(H)t(H^c) &=  \frac{t^2(H) + t(H)^2}{2} + \frac{t^2(H^c) + t(H^c)^2}{2}+ t(H)t(H^c) \\
&= \frac{t^2(Y) +t(Y)^2}{2} \\
&= t(Y) V(Y).
\end{align*}
\end{proof}




We now introduce a family of search strategies, indexed by the parameter $\alpha$, where $0 \le \alpha \le 1$.

\begin{definition} 
Let $s^*=s^*(\alpha)$ be the search that with probability~$\alpha$ follows $s_H$ then follows $s_{H^c}$, and with probability $1-\alpha$ follows $s_{H^c}$ then follows $s_H$. 
\label{def:s_box}
\end{definition}

The next lemma shows that the consistency and robustness of strategies of the form $s^*(\alpha)$ satisfy a linear relation when $\alpha$ is at least some constant $\alpha^*$. This is the last, and most crucial part of the lemma. The first part gives expressions for the expected search time of $s^*(\alpha)$ for boxes in both $H$ and $H^c$. The second part of the lemma shows that for $\alpha=\alpha^*$, the consistency and robustness are both equal to the value $V(Y)$ of the game. It is clear that for any smaller $\alpha$, the strategy $s^*(\alpha)$ cannot be Pareto optimal, since both the consistency and the robustness will be equal to the expected search time of boxes in $H$, and this will be greater than $V(Y)$.

\begin{lemma} \label{lem:P-opt}
Let 
$\alpha^*= 1-\frac{V(Y)-V(H)}{t(H^c)}$. 
Then 
\begin{enumerate}
\item[(i)] For $i \in H, j \in H^c$ and any $0 \le \alpha \le 1$,
\begin{align*}
u(s^*(\alpha),i)&=V(H) + (1-\alpha)t(H^c) \text{ and } \\
u(s^*(\alpha),j) &=  V(H^c) + \alpha t(H);
\end{align*}
\item[(ii)] 
$C(s^*(\alpha^*))= R(s^*(\alpha^*)) = V(Y)$;
\item[(iii)] For any $\alpha^* \le \alpha \le 1$, it holds that 
$
t(H) C(s^*(\alpha)) + t(H^c) R(s^*(\alpha)) = t(Y) V(Y).
$
\end{enumerate}
\end{lemma}
\begin{proof}
Part (i) is immediate, by definition of $s^*(\alpha)$. Indeed, for boxes $i \in H$, the strategy $s^*(\alpha)$ searches all the boxes in $H^c$ first with probability $1-\alpha$, taking time $t(H^c)$, then it performs an optimal search of $H$, which finds the target in additional time $V(H)$, by Theorem~\ref{thm:box}. The expression for $u(s^*(\alpha),j)$ is derived similarly.

For part (ii), we first calculate $u(s^*(\alpha^*),i)$ and $u(s^*(\alpha^*),j)$ for $i \in H$ and $j \in H^c$ using part (i) and the definition of $\alpha^*$.
\begin{align*}
u(s^*(\alpha^*), i) &= V(H) + \left( \frac{V(Y)-V(H)}{t(H^c)} \right) t(H^c) = V(Y) \text{ and} \\
u(s^*(\alpha^*), j) &= V(H^c) + \left( 1- \frac{V(Y)-V(H)}{t(H^c)} \right) t(H) = \frac{V(Y)t(Y)-V(Y)t(H)}{t(H^c)} = V(Y).
\end{align*}
where the penultimate equality uses~(\ref{eq:V}) and the final equality uses $t(Y)=t(H)+t(H^c)$. It follows that the consistency and robustness of $s^*(\alpha^*)$ are both equal to $V(Y)$.

Finally, for part (iii), we note that $u(s^*(\alpha),i)$ is decreasing in $\alpha$ for $i \in H$, and $u(s^*(\alpha),j)$ is increasing in $\alpha$ for $j \in H^c$. By part (ii), if follows that for all $\alpha^* \le \alpha \le 1$, we have $u(s^*(\alpha),i) \le u(s^*(\alpha),j)$ for $i \in H, j \in H^c$, so the consistency of $s^*(\alpha)$ is equal to $u(s^*(\alpha),i)$ and the robustness of $s^*(\alpha^*)$ is equal to $u(s^*(\alpha),j)$. Therefore, by part (i), for $\alpha^* \le \alpha \le 1$,
\begin{align*}
t(H) C(s^*(\alpha)) + t(H^c) R(s^*(\alpha)) &= t(H)(V(H) + (1-\alpha)t(H^c)) + t(H^c)(V(H^c) + \alpha t(H)) \\
&= t(Y)V(Y),
\end{align*}
by~(\ref{eq:V}).
\end{proof}

It is interesting to point out that the strategy $s^*(\alpha^*)$ is a min-max search strategy for the standard box search problem (with no prediction), since every box is found in expected time $V(Y)$.

We now prove a tight lower bound, by showing the the linear relation satisfied by the robustness and consistency or $s^*(\alpha)$ holds as an inequality for any search $s$.

\begin{lemma} \label{lem:box-lb}
The robustness and consistency of any search strategy $s$ satisfy
\[
t(H) C(s) + t(H^c) R(s) \ge t(Y) V(Y).
\]
\end{lemma}
\begin{proof}
Let $s$ be an arbitrary mixed search strategy. By Lemma~\ref{lem:hiding}, we must have $u(s, h^*_Y) = V(Y)$. We will now write $u(s, h^*_Y)$ in a different way, using the fact that the distribution $h^*_{H}$ is equal to the distribution $h^*_{Y}$ conditional on the target lying in $H$, and similarly for $h^*_{H^c}$. So,
\[
V(Y)=u(s, h^*_Y) = h^*_Y(H) u(s,h^*_H) + h^*_Y(H^c) u(s,h^*_{H^c}).
\]
By the definitions of consistency and robustness, we must have $C(s) \ge u(s,h^*_H) $ and $R(s) \ge u(s,h^*_{H^c})$. Combining with the  identities $h^*_Y(H) = t(H)/t(Y)$ and $h^*_Y(H^c) = t(H^c)/t(Y)$, we obtain
\[
V(Y) \le \frac{t(H)}{t(Y)} C(s) +  \frac{t(H)^c}{t(Y)} R(s).
\]
The lemma follows.
\end{proof}

Combining the above lemmas, we can characterize the Pareto frontier of the game.
\begin{theorem} 
For a given prediction $H$, the Pareto frontier $P$ 
is given by the line segment
\[
P=\bigl\{ (c,r): t(H) c + t(H^c) r = t(Y) V(Y),~ V(H) \le c \le V(Y)  \bigl\}.
\]
\end{theorem}


\begin{example}
Figure~\ref{fig:tradeoff1} illustrates the Pareto frontier for $n=2$, $t_1=t_2=1$, and $H=\{1\}$. In this case, the consistency must lie between $V(H)=1$ and $V(Y)=1.5$. We mark both the lines $c=1$ and $c=1.5$ on the graph with a dotted line. 
\label{ex:example1}
\end{example}

\begin{figure}[ht!]
\centering
\includegraphics[scale=0.2]{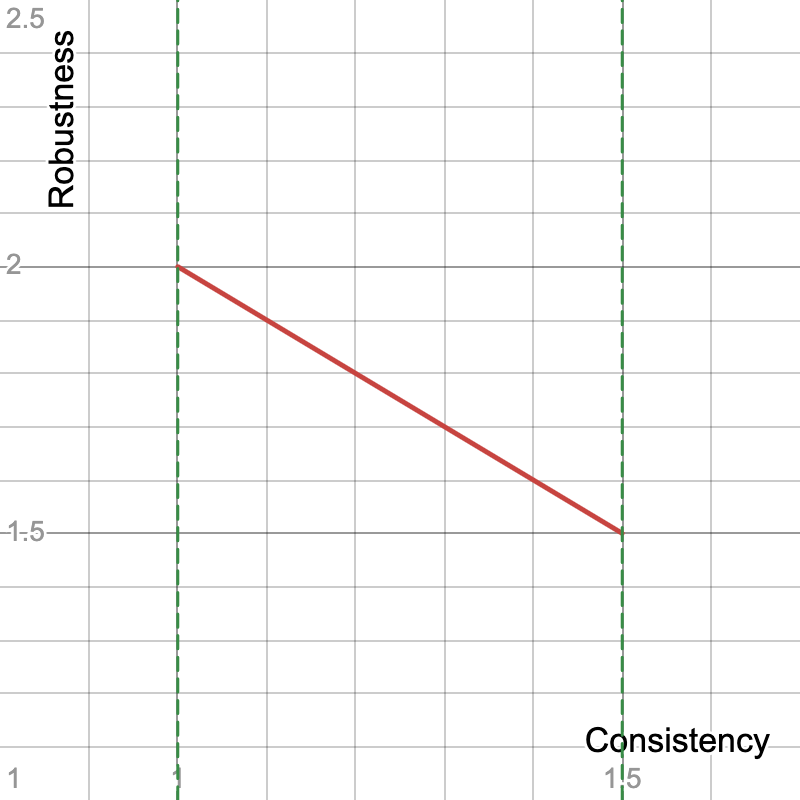}
\caption{An illustration of the Pareto frontier for the search game with perfect detection.}
\label{fig:tradeoff1}
\end{figure}

\subsection{Box search with imperfect detection} \label{sec:box-infinite}

We now turn to the second, and more general variant of our game. Here, each time a box containing the Hider is searched, the Hider is independently found with probability $q$, where $0<q<1$. The Hider's strategy set $Y$ remains the same, but the Searcher's strategy set $X$ is now the set $Y^\infty$ of all {\em infinite} sequences of boxes. We allow the Searcher to randomize between any {\em finite} set of strategies 
from~$X$. A solution to the standard game was given in~\cite{bui2024optimal}, and we state the result below. The result provides a connection between the two standard games, with perfect and imperfect detection.

\begin{theorem}[\cite{bui2024optimal}] 
Given an optimal strategy for box search with perfect detection, and any given permutation $\sigma$ of $Y$, let 
$\theta(\sigma)$ denote the probability
with which the strategy chooses $\sigma$.
Then, for box search with imperfect detection, the strategy that repeats $\sigma$ indefinitely with probability $\theta(\sigma)$ is optimal for the Searcher. The expected search time $V_q$ of this strategy (equal to the value of the game) is
\[
V_q(Y)=\frac{t(Y)}{q}-\frac{\sum_{1\le i < j \le n} t_i t_j}{t(Y)}.
\]
\end{theorem}
Note that 
$V_1(Y) = \frac{t(Y)^2- \sum_{1\le i < j \le n} t_i t_j}{t(Y)} 
= \frac{t^2(Y) + t(Y)^2}{2t(Y)}$,
and we recover the value of the game for perfect detection given in~(\ref{eq:box-value}). Thus, we can write the value of the game for imperfect detection as
\begin{align}
V_q(Y) &= \frac{(1-q)t(Y)}{q} + t(Y)-\frac{\sum_{1\le i < j \le n} t_i t_j}{t(Y)} 
= V_1(Y) + \frac{(1-q)t(Y)}{q} . \label{eq:Vq}
\end{align}

As for perfect detection, a prediction is given by a proper subset $H \subset [n]$, and we assume that $H$ is some such fixed subset.
We will use a simple lemma that is analogous to Lemma~\ref{lem:V-identity}.
\begin{lemma}\label{lem:Vq-identity}
For any $H \subseteq Y$,
\begin{align} 
t(Y) V_q(Y) &= t(H) V_q(H) + t(H^c) V_q(H^c) + \frac{2-q}{q}t(H)t(H^c), \label{eq:Vq-identity}
\end{align}
\end{lemma}

\begin{proof}
Using~(\ref{eq:Vq}),~(\ref{eq:V}) and the identity $t(Y)=t(H)+t(H^c)$, we can write the left-hand side of~(\ref{eq:Vq-identity}) as
\begin{align*}
t(Y) V_q(Y) &= t(H)\left(V_1(H) + \frac{1-q}{q}t(H)\right) + t(H^c)\left(V_1(H^c) + \frac{1-q}{q}t(H^c)\right) + \frac{2-q}{q}t(H)t(H^c) \\
&= t(H) V_q(H) + t(H^c) V_q(H^c) + \frac{2-q}{q}t(H)t(H^c),
\end{align*}
again using~(\ref{eq:Vq}), with $Y$ replaced with $H$ and $H^c$ respectively. 
\end{proof}

 We now define a set of search strategies which we will show are Pareto optimal.
 Recall the definition of $s_H$, $s_{H^c}$, as given in Section˜\ref{sec:box-finite}.
 
\begin{definition} For $k=0,1,\ldots$, let $s^k$ be the search which begins by performing the search $s_H$ a total of $k$ times before repeatedly alternating between $s_{H^c}$ and $s_H$. 
For $0 \le \beta \le 1$, let $s^k(\beta)$ be the search that chooses $s^k$ with probability $1-\beta$ and chooses $s^{k+1}$ with probability $\beta$. 
\end{definition}

Our next lemma mirrors Lemma~\ref{lem:P-opt} from Subsection~\ref{sec:box-finite}. Part (iii) shows how the consistency and robustness of the strategies $s^k(\beta)$ are related. For each $k=0,1,\ldots$, the consistency and robustness satisfy a different, and more complex linear relation.

\begin{lemma}\label{lem:tradeoff}
Let
$
 \beta^* \equiv \frac{1}{q}-\frac{V_q(Y)-V_q(H)}{t(H^c)}.
$
Then
\begin{enumerate}
\item[(i)] For $i \in H, j \in H^c$, $k=0,1,\ldots$ and any $0 \le \beta \le 1$,
\begin{align*}
u(s^k(\beta),i)&= V_q(H)  + \left( \frac{1}{q} - \beta \right)  (1-q)^k t(H^c) \text{ and } \\
u(s^k(\beta),j) &=  V_q(H^c) + \left(\beta + k+ \frac{1-q}{q} \right)t(H);
\end{align*}
\item[(ii)] $C(s_0(\beta^*))= R(s_0(\beta^*)) = V_q(Y)$;
\item[(iii)] for $k = 1,2,\ldots$ and  $0 \le \beta \le 1$, and $k=0$ and $\beta^* \le \beta \le 1$,
\begin{align}
\frac{t(H)}{(1-q)^k} C(s^k(\beta)) + t(H^c) R(s^k(\beta)) = t(Y)V_q(Y) + \frac{1-(1-q)^k}{(1-q)^k} t(H)V_q(H) + kt(H)t(H^c), \label{eq:C-R}
\end{align}
\end{enumerate}
\end{lemma}

\begin{proof}
Starting with part (i), under $s^k$, the expected search time of an element $i \in H$ is
\[
u(s^k, i)  = V_q(H)  +   \sum_{\ell=k}^\infty (1-q)^{\ell}t(H^c)  = V_q(H) + \frac{(1-q)^k}{q} t(H^c).
\]
Therefore, the expected search time of $i\in H$ under $s^k(\beta)$ is
\begin{align*}
u(s^k(\beta), i) & = V_q(H)  + (1-\beta)\frac{(1-q)^k}{q} t(H^c) + \beta \frac{(1-q)^{k+1}}{q} t(H^c) \\
&= V_q(H)  + \left( \frac{1}{q} -  \beta \right)  (1-q)^k t(H^c) .
\end{align*}

The expected search time of an element $j \in H^c$ is
\begin{align*}
u(s^k, j) & = V_q(H^c) + \left(k+ \sum_{\ell=1}^{\infty} (1-q)^\ell \right)t(H) =V_q(H^c) + \left(k+ \frac{1-q}{q} \right)t(H).
\end{align*}
Therefore, the expected search time of $j \in H^c$ under $s^k(\alpha)$ is
\begin{align*}
u(s^k(\beta), j) & = V_q(H^c) + (1-\beta) \left(k+ \frac{1-q}{q} \right)t(H) + \beta  \left(k+1+ \frac{1-q}{q} \right)t(H) \\
&= V_q(H^c) + \left(\beta + k+ \frac{1-q}{q} \right)t(H).
\end{align*}

For part (ii), we use part (i) and the definition of $\beta^*$ to calculate $u(s_0(\beta^*),i)$ for $i \in H$.
\[
u(s_0(\beta^*),i) = V_q(H)  + \left( \frac{1}{q} - q \left(\frac{1}{q}-\frac{V_q(Y)-V_q(H)}{t(H^c)} \right) \right) t(H^c)  = V_q(Y),
\]
For $j \in H^c$, we have
\begin{align*}
u(s_0(\beta^*),j) &=  V_q(H^c) + \left(\frac{1}{q}-\frac{V_q(Y)-V_q(H)}{t(H^c)} + \frac{1-q}{q} \right)t(H) \\
&= \frac{1}{t(H^c)} \left(t(H) V_q(H) + t(H^c)V_q(H^c) + \frac{2-q}{q} t(H)t(H^c) - t(H) V_q(Y) \right) \\
&= V_q(Y),
\end{align*}
using~(\ref{eq:Vq-identity}) and $t(Y)=t(H)+t(H^c)$.  It follows that the consistency and robustness of $s^k(\beta^*)$ are both equal to $V_q(Y)$.

Finally, for part (iii), we note that $u(s^k(\beta),i)$ is decreasing in $\beta$ for $i \in H$, and $u(s^k(\beta),j)$ is increasing in $\beta$ for $j \in H^c$. Also, $s^k(1)=s^{k+1}(0)$ for all $k=0,1,\ldots$. By part (ii), if follows that for all $\beta^* \le \beta \le 1$, we have $u(s_0(\beta),i) \le u(s_0(\beta),j)$ for $i \in H, j \in H^c$; also for all $0 \le \beta \le 1$ and $k=1,2,\ldots$, we have $u(s^k(\beta),i) \le u(s^k(\beta),j)$ for $i \in H, j \in H^c$. So, the consistency of $s^k(\beta)$ is equal to $u(s^k(\beta),i)$ and the robustness of $s^k(\beta)$ is equal to $u(s(\beta),j)$ for $\beta^* \le \beta \le 1$ and $k=0$, and for $0 \le \beta \le 1$ and $k=1,2,\ldots$. Therefore, by part (i), for $\beta^* \le \beta \le 1$ and $k=0$, and for $0 \le \beta \le 1$ and $k=1,2,\ldots$, we have
\begin{align*}
\frac{t(H)}{(1-q)^k} C(s^k(\beta)) + t(H^c) R(s^k(\beta)) &= \frac{t(H)}{(1-q)^k} \left( V_q(H)  + \left( \frac{1}{q} -  \beta \right)  (1-q)^k t(H^c) \right) \\
&\quad + t(H^c) \left( V_q(H^c) + \left(\beta + k+ \frac{1-q}{q} \right)t(H)  \right) \\
&=t(Y)V_q(Y) + \frac{1-(1-q)^k}{(1-q)^k} t(H)V_q(H) + kt(H)t(H^c),
\end{align*}
by~(\ref{eq:Vq-identity}). This completes the proof.
\end{proof}

Proving a tight lower bound is intricate. For the finite game of Subsection~\ref{sec:box-finite}, we gave a lower bound on the linear combination of consistency and robustness by exploiting the fact that every search strategy had the same expected search time against the hiding strategy $h_Y^*$. For the game of this section, this is not true, and we will prove a lower bound corresponding to~(\ref{eq:C-R}) using a family of Hider strategies $h^k,~k=0,1,\ldots$. We will show that a best response to $h^k$ is a search that cycles $k$ times through the boxes in $H$ before repeatedly cycling through all the boxes in $Y$.

Before formally defining $h^k$, we need to give some background on finding best responses to fixed hiding distributions, which will also motivate the choice of $h^k$. Suppose the probability the hider is in box $j$ is known to be $h(j)$ for $j=1,\ldots,n$. Then the problem of finding a search with minimum expected search time is well understood (see for example~\cite{bui2024optimal}). The optimal strategy can be found by first choosing a box $j$ (of which there may be several) that maximizes the ratio $h(j) /t_j$. Subsequently, every time a box is searched, the hiding probabilities are updated using Bayes' law, and the next box to be searched is chosen by repeating the calculation. Clearly, if the hiding distribution is $h^*_Y$, then the ratios $h^*_Y(j)/t_j$ are all equal, so it is optimal to start by searching all $n$ boxes exactly once in any order. The hiding distribution then returns to $h^*_Y$, after Bayesian updating, so an optimal search strategy repeatedly cycles through the boxes in any order.

With this in mind, we formally define the hiding distribution $h^k$ as 
\[
h^k(j) = \begin{cases}
\frac{\lambda_k t_j}{(1-q)^k} \text{ for $j \in H$ and} \\
\lambda_k t_j \text{ for } j \in H^c,
\end{cases}
\]
where $\lambda_k = (t(H)/(1-q)^k + t(H^c))^{-1}$ is a normalizing factor that ensures $h^k$ is a probability distribution. The strategy $h^k$ is designed so that  any optimal response to it begins with $k$ ``rounds'' of searching all the boxes in $H$ in any order, after which, updating the hiding probabilities using Bayes' law, the hiding distribution becomes $h^*_Y$. Thus, an optimal response to $h^k$ continues by repeatedly cycling through all the boxes of $Y$ in any order after those first $k$ rounds. We can calculate the expected search time of an optimal response to $h^k$.

\begin{lemma}\label{lem:hk-lb}
The minimum expected search time against the hiding distribution $h^k$ is given by
\[
\lambda \left( \frac{1-(1-q)^k}{(1-q)^k} t(H)V_q(H) -  k t(H)t(H^c)+ t(Y)V_q(Y) \right).
\]
\end{lemma}

\begin{proof}
Let $s$ be a best response to $h^k$, so that $s$ cycles through all the boxes in $H$ in some order $k$ times, before repeatedly cycling through all the boxes in $Y$ in any order. Then
\[
u(s,h^k) = h^k(H) \left( (1-(1-q)^k)V_1(H) + t(H) \sum_{j=1}^{k-1} q(1-q)^j j \right) + (1-h^k(H)(1-(1-q)^k))(kt(H)+V_q(Y)).
\]
Using the definition of $h^k$ and the following identity
\[
\sum_{j=1}^{k-1} q(1-q)^j j \equiv (1-(1-q)^k) \frac{1-q}{q}  - k(1-q)^k,
\]
the expected search time $u(s,h^k) $ can be rewritten
\begin{align*}
u(s,h^k) &= \frac{\lambda_k t(H)}{(1-q)^k} \left( (1-(1-q)^k) \left(V_1(H)  + t(H) \frac{1-q}{q} \right) - k(1-q)^k t(H) \right) \\
& \quad + \lambda_k t(Y) \left( kt(H)+V_q(Y) \right) \\
&= \lambda_k \left( \frac{1-(1-q)^k}{(1-q)^k} t(H)V_q(H) -  k t(H)t(H^c)+ t(Y)V_q(Y) \right),
\end{align*}
using~(\ref{eq:Vq}) and $t(Y)=t(H)+t(H^c)$.
\end{proof}

We will now use the hiding distribution $h^k$ to show that (\ref{eq:C-R}) provides the optimal consistency/robustness tradeoff.

\begin{lemma}\label{lem:tradeoff2}
For any search $s$ and any $k=0,1,\ldots$, we have
\[
\frac{t(H)}{(1-q)^k} C(s) + t(H^c) R(s) \ge t(Y)V_q(Y) + \frac{1-(1-q)^k}{(1-q)^k} t(H)V_q(H) + kt(H)t(H^c).
\]
\end{lemma}

\begin{proof}[Proof of Lemma~\ref{lem:tradeoff2}]
 Let $s$ be an arbitrary mixed search strategy. By Lemma~\ref{lem:hk-lb}, we must have 
\[
u(s, h^k) \ge \lambda_k \left( \frac{1-(1-q)^k}{(1-q)^k} t(H)V_q(H) -  k t(H)t(H^c)+ t(Y)V_q(Y) \right).
\] 
We will now write $u(s, h^k)$ in a different way, using the fact that the distribution $h^*_{H}$ is equal to the distribution $h^k$ conditional on the hider lying in $H$, and similarly for $h^*_{H^c}$. So,
\begin{align*}
u(s, h^k) &= h^k(H) u(s,h^*_H) + h^k(H^c) u(s,h^*_{H^c}) \\
&= \frac{\lambda_k t(H)}{(1-q)^k} u(s,h^*_H)+ \lambda_k t(H^c) u(s,h^*_{H^c}) \\
&\le \frac{\lambda_k t(H)}{(1-q)^k} C(s)+ \lambda_k t(H^c) R(s),
\end{align*}
by the definitions of consistency and robustness. The lemma follows.
\end{proof}

Combining these lemmas, we can characterize the Pareto frontier.
\begin{theorem}
For a given prediction $H$, the Pareto frontier $P$ 
is given by 
\begin{align*}
P &= \bigcup_{k=1}^\infty \biggl\{ (c,r): \frac{t(H)}{(1-q)^k} c+ t(H^c) r = t(Y)V_q(Y) + \frac{1-(1-q)^k}{(1-q)^k} t(H)V_q(H) + kt(H)t(H^c), \\
&\quad v_{k+1} \le c \le v_k \biggl\} \cup \bigl\{ (c,r):  t(H) c + t(H^c) r = t(Y)V_q(Y), v_1 \le c \le V_q(Y) \bigl\} ,
\end{align*}
where $v_k = V_q(H)  + \frac{(1-q)^k}{q} t(H^c)$ for $k=0,1,\ldots$.
\end{theorem}


\begin{example}
Figure~\ref{fig:tradeoff} illustrates part of the Pareto frontier for $n=2$, $t_1=t_2=1$, $q=1/2$ and $H=\{1\}$. We only include points on the Pareto frontier given by strategies $s^k(\beta)$ for $k \le 7$. For this example, the consistency must lie between $V(H)=2$ and $V(Y)=3.5$. We mark both the lines $c=2$ and $c=3.5$ on the graph with a dotted line. Observe that as the consistency tends to 2, the robustness tends to infinity. 
\label{ex:example2}
\end{example}

\begin{figure}[ht!]
\centering
\includegraphics[scale=0.2]{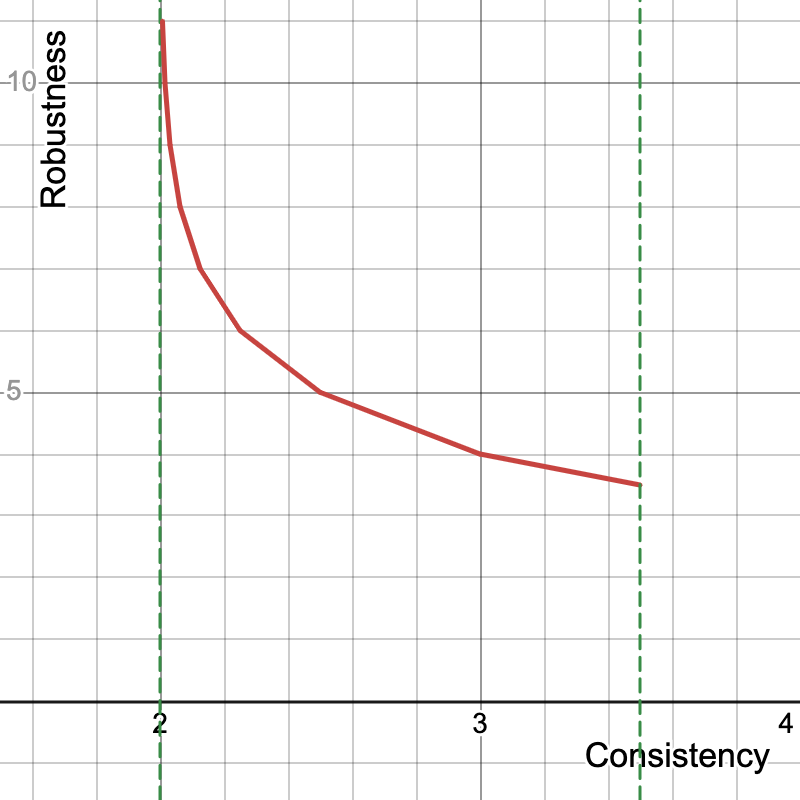}
\caption{An illustration of the Pareto frontier for the search game with imperfect detection.}
\label{fig:tradeoff}
\end{figure}

\section{A General Approach to Characterizing the Pareto Frontier} 
\label{sec:general}

In the search games studied in the previous sections, we made use of explicit characterizations of strategies that have consistency or robustness bounded by some given value. In general, however, such a useful characterization may not be readily available for a given game, as it happens, e.g., in linear search, as we will discuss in Section~\ref{sec:linear.search}. To address this complication, we will now present a general approach to finding the Pareto frontier, which holds not only for search games, but for arbitrary two-player zero-sum games $G$. As usual, $X$ denotes the strategy set of Player 1 (the minimizer) and $Y$ denotes the strategy set of Player 2 (the maximizer). Recall that $\Delta(X)$ denotes the set of mixed strategies of Player 1. We consider the hint $H \subset X$ to be fixed.

We will denote by $A$ the set of consistency-robustness pairs $\{(C(x),R(x)): x \in \Delta(X)\}$, and for $\boldsymbol{\lambda}=(\lambda_1,\lambda_2) \in \mathbb{R}^2_{\ge 0}$, let $P_{\boldsymbol{\lambda}}$ denote the set of pairs $(C(\bar{x}),R(\bar{x})) \in A$ such that $\bar{x}$ minimizes $\lambda_1C(x)+\lambda_2 R(x)$. (We assume that $P_{\boldsymbol{\lambda}}$ is well defined.) We would like to show that the Pareto frontier of $A$ is equal to $\cup_{\boldsymbol{\lambda} \in \mathbb{R}^2_{\ge 0} \setminus \{0\}} P_{\boldsymbol{\lambda}}$, so that it can be found by minimizing all linear combinations $\lambda_1C(x)+\lambda_2 R(x)$. If $A$ were convex, this would follow immediately by applying the well-known concept of {\em scalarization} (for example, see \cite{boyd2004convex}). However, in general, $A$ may not be convex. We prove the following lemma by applying scalarization to the convex hull $\text{conv}(A)$ of $A$.

\begin{lemma} \label{lem:PF}
The Pareto frontier of $A$ is precisely $\cup_{\boldsymbol{\lambda} \in \mathbb{R}^2_{\ge 0} \setminus \{0\}} P_{\boldsymbol{\lambda}}$.
\label{lemma:pareto.characterization}
\end{lemma}

\begin{proof}
We first show that for every point $(c,r)\in \text{conv}(A)$, there is a point $(C(x^*),R(x^*)) \in A$, for some $x^* \in \Delta(X)$ such that $(C(x^*),R(x^*)) \le (c,r)$. Indeed, such a point $(c,r) \in \text{conv}(A)$ can be written 
as 
\[
(c,r) = \alpha(C(x_1),R(x_1)) + (1-\alpha)(C(x_2), R(x_2)),
\]
for some strategies $x_1,x_2 \in \Delta(X)$ and some $0 \le \alpha \le 1$. Let $x^*$ be the strategy that chooses $x_1$ with probability $\alpha$ and $x_2$ with probability $1-\alpha$. Then
\begin{align*}
C(x^*) &  = \sup_{y \in H} u(x^*,y) \\
&= \sup_{y \in H} (\alpha u(x_1,y) + (1-\alpha) u(x_2,y)) \\
&\le \alpha \sup_{y \in H} u(x_1,y)  + (1-\alpha) \sup_{y \in H} u(x_2,y) \\
&= \alpha C(x_1) + (1-\alpha)C(x_2) \\
&=c.
\end{align*}
Similarly, $R(x^*) \le \alpha R(x_1) + (1-\alpha)R(x_2) = r$. Hence, $(C(x^*),R(x^*)) \le (c,r)$.

We will now show that $A$ has the same Pareto frontier as $\text{conv}(A)$. Let $(c,r)$ be a Pareto optimal point of $\text{conv}(A)$. Then, by our observation above, there is some point $(C(x^*),R(x^*)) \in A \subseteq \text{conv}(A)$ such that $(C(x^*),R(x^*)) \le (c,r)$. By the Pareto optimality of $(c,r)$, we must have $(c,r) = (C(x^*),R(x^*))$, so $(c,r) \in A$ and must be Pareto-optimal in $A$. 

Suppose now that some point $(c,r)$ is not Pareto optimal in $\text{conv}(A)$. Then there exists 
another point $(c',r') \le (c,r)$ such that $(c',r') \neq (c,r)$. By our observation, there is a point $(C(x^*),R(x^*)) \le (c',r')$, so that this point Pareto dominates $(c,r)$, and so $(c,r)$ cannot be Pareto optimal in $A$. 
We conclude that the Pareto frontier of $A$ and $\text{conv}(A)$ are equal. 

To complete the proof, we must show that $\cup_{\boldsymbol{\lambda} \in \mathbb{R}^2_{\ge 0} \setminus \{0\}} P_{\boldsymbol{\lambda}}$ contains all the points on the Pareto frontier of $A$. Suppose that some point $(C(\bar{x}),R(\bar{x}))$ is on the Pareto frontier of $A$. Then it is also on the Pareto frontier of $\text{conv}(A)$. Since $\text{conv}(A)$ is convex, $(C(\bar{x}),R(\bar{x}))$ can be found be applying scalarization to $\text{conv}(A)$, so that $(C(\bar{x}),R(\bar{x}))$ minimizes $\lambda_1 c + \lambda_2 r$  over all pairs $(c,r) \in \text{conv}(A)$, for some $\lambda_1,\lambda_2$. But in that case, $(C(\bar{x}),R(\bar{x}))$ must also minimize $\lambda_1 C(x) + \lambda_2 R(x)$ over $x \in \Delta(X)$, and therefore lies in $P_{(\lambda_1,\lambda_2)}$.
\end{proof}

The above lemma shows that a point on the Pareto frontier can be found be solving a minimization problem. We now go on to show that the solution of this problem coincides with the solution of the following auxiliary game.


\begin{definition}
Let $G$ be a zero-sum game with Player 1 strategies $X$, Player 2 strategies $Y$, and payoff function $u$. Let $H \subset Y$ be a prediction. Given $\boldsymbol{\lambda} \in \mathbb{R}^2_{\ge 0}$, define the zero-sum game $G_{\boldsymbol{\lambda}}$
as follows: Player~1's pure strategy set is $X$, Player~2's pure strategy set is $H \times Y$, and the payoff is given by 
\[
v(x,(y_1,y_2)) = \lambda_1 u(x,y_1)+\lambda_2 u(x,y_2).
\]
\label{def:aux.game}
\end{definition}
\begin{lemma}
If the game $G_{\boldsymbol{\lambda}}$ has a value and optimal (min-max) strategies for Player~1, then the set of optimal strategies for Player~1 in $G_{\boldsymbol{\lambda}}$ is precisely $P_{\boldsymbol{\lambda}}$.
Specifically, it is sufficient to find optimal strategies in the games $G_{\boldsymbol{\lambda}}$ for $\boldsymbol{\lambda}= (0,1)$ and $\boldsymbol{\lambda} = (1,\lambda)$, where $\lambda \ge 0$.
\label{lem:pareto.game}
\end{lemma}
\begin{proof}
The value $V$ of $G_{\boldsymbol{\lambda}}$ is given by
\begin{align*}
V &= \min_{x \in \Delta(X)} \max_{(y_1,y_2) \in H \times Y} v(x,(y_1,y_2)) \\
& =  \min_{x \in \Delta(X)} \max_{(y_1,y_2) \in H \times Y}  \lambda_1 u(x,y_1)+\lambda_2 u(x,y_2) \\
& =  \min_{x \in \Delta(X)} \left( \max_{y_1 \in H} \lambda_1 u(x,y_1)+\lambda_2 \max_{y_2 \in Y} u(x,y_2) \right)\\
& = \min_{x \in \Delta(X)} (\lambda_1C(x) + \lambda_2 R(x)).
\end{align*}
Hence, a strategy $x$ is optimal if and only if it lies in $\arg \min \{\lambda_1 C(x)+\lambda_2 R(x): x \in \Delta(X)\} \equiv P_{\mathbf{\lambda}}$.

By Lemma~\ref{lem:PF}, the Pareto frontier is equal to the set of values of games $G_{\boldsymbol{\lambda}}$ for $\boldsymbol{\lambda} \in \mathbb{R}^2_{\ge 0}$.  By rescaling, it is clearly sufficient to consider only the cases $\boldsymbol{\lambda}=(0,1)$ or $\boldsymbol{\lambda}=(1,\lambda_2)$, where $\lambda_2 \ge 0$. Note that for $\boldsymbol{\lambda}=(0,1)$, the game $G_\lambda$ reduces to the standard game $G$.
\end{proof}




\section{Searching in the Infinite Line}
\label{sec:linear.search}

In this section, we focus on the linear search problem, which is one of the fundamental search games in unbounded spaces; see Chapters 8 and 9 in~\cite{searchgames} for many variants of this classic game. Another important aspect of this problem is that, unlike box search, it is not obvious how to characterize strategies of bounded consistency/robustness. We will thus use it so as to illustrate our approach of Section~\ref{sec:general}.

In the standard version of the game (with no predictions) the search space consists of an infinite line with an origin $O$, and the Searcher (initially placed at $O$) must locate an immobile Hider that lies at some unknown point of the line. 
A pure strategy for the Searcher can be defined as a biinfinite sequence\footnote{This definition allows for infinitesimal oscillations around the origin, which is one way to guarantee that the game has bounded value. Another way is to require that the Hider is at least a unit distance away from $O$. We adopt the former, which is common in the study of search games and leads to simpler and cleaner expressions; see also the discussion in~\cite{demaine:turn}.} of the form $(x_i,\tau_i)_{i=-\infty}^{+\infty}$, where $x_i$ is the length to which the line is searched
in each iteration, and $\tau_i \in \{0,1\}$ is the direction of the search (where $0/1$ signifies left and right directions respectively). Namely, in iteration $i$, the Searcher starts from $O$, explores the half line $\tau_i$ to distance $x_i$ and returns to $O$. Without loss of generality, we will assume that $\tau_i$ and $\tau_{i+1}$ are of opposite parities (i.e., the Searcher does not visit the same halfline in consecutive iterations) and $x_{i+2}>x_i$ (i.e., the Searcher always explores a new part of the line, in each iteration). A mixed strategy for the Searcher is, as usual, a randomized choice of pure strategies. One way to specify a mixed strategy for the Searcher is to choose a sequence $(x_i,\tau_i)_{i=-\infty}^{+\infty}$, where the $x_i$ are random variables such that $x_{i+2}>x_i$. The set $Y$ of pure strategies for the Hider is $\mathbb{R}\setminus \{0\}$, so that a mixed strategy is a probability distribution on $\mathbb{R}\setminus \{0\}$. Here, we use the convention that positive numbers refer to the right halfline, whereas negative numbers refer to the left halfline. 

Given a strategy $y \in Y$ for the Hider, and a pure strategy $S$ for the Searcher, we define the  payoff $u(S,y)=\frac{c(S,y)}{|y|}$, where $c(S,y)$ denotes the {\em cost} at which a Searcher that follows $S$ locates $y$, and $|y|$ is the distance of the Hider from $O$. In other words, if the search $S$, given by the sequence $(x_i,\tau_i)_{i=-\infty}^{+\infty}$ first reaches the point $y$ during the $j$th iteration, then 
\[
c(S,y) = y+2\sum_{i=-\infty}^{j-1} x_i.
\]
Finding the value of the game amounts to minimizing, over mixed strategies $s$, the worst-case expected {\em normalized} cost
\begin{equation}
\comp(s)= \sup_y \frac{\mathbb{E}[c(s,y)]}{|y|},
\end{equation}
where the expectation is taken with respect to $s$. As first observed in~\cite{beck:yet.more}, this normalization is essential for the game to have a value. Following the TCS terminology, we refer to the 
worst-case expected normalized cost of a search strategy as its {\em competitive ratio}; hence finding the value of the game is equivalent to finding the best (minimum) competitive ratio among all possible strategies. 
In~\cite{searchgames}  as well as~\cite{ray:2randomized}, it was shown that the best randomized competitive ratio equals 
\begin{equation}
\rho^*= 1+\frac{\alpha^*}{\ln \alpha^*},
\label{eq:randomized.cr}
\end{equation} 
where $\alpha^*$ is the minimizer of the function $\frac{1+\alpha}{\ln \alpha}$, with $\alpha>1$. Numerically, $\alpha^* \approx 3.59$ hence the optimal competitive ratio is approximately 4.59.

In the remainder of this section, we study the linear search game in a setting in which the Searcher is enhanced with a prediction $H \subset \mathbb{R}$ about the location of the Hider. Specifically, we will assume that $H$ can be either the set of positive reals, or the set of negative reals, hence $H$ defines a {\em directional} prediction. By symmetry,  we will assume, without loss of generality, that $H$ is the set of positive numbers, or equivalently, that the Hider is at the right of $O$. 
Following our definitions from Section~\ref{sec:preliminaries}, the robustness of a search strategy $s$ is the competitive ratio of $s$, whereas the consistency of $s$ is the maximum expected normalized search cost of $s$, where the maximum is taken over only hiding strategies $y>0$. 

We begin with an upper bound obtained by an explicit search strategy. 

\begin{definition} \label{def:upper.linear}
Given $\alpha>1$ and $\mu \in [0,1]$, define $s_{\alpha, \mu}=(x_i,\tau_i)_{i=-\infty}^{\infty}$ as the mixed strategy whose search lengths are such that $\tau_0=1$ (i.e., even indexed iterations search the predicted halfline); moreover, let $u$ be chosen uniformly at random in the interval $[0,2]$, then the search lengths of $s$ are defined as $x_i=\alpha^{i+u}$, if $i$ is even, and $x_i=\mu \alpha^{i+u}$, if $i$ is odd.
\label{def:s}
\end{definition}

Informally, $s_{\alpha, \mu}$ is a ``biased'' geometric strategy with base $\alpha$ and a random offset $\alpha^u$, and which invests less cost in searching the non-predicted halfline, by a factor $\mu$. Note that for $\mu=1$, we recover the strategy of optimal competitive ratio.  We can easily analyse this family of  strategies:

\begin{proposition}
For any $\alpha>1$ and $\mu\in (0,1]$, $s_{\alpha, \mu}$ has consistency at most $1+\frac{1+\mu\alpha}{\ln \alpha}$ and robustness at most $1+\frac{1+\frac{\alpha}{\mu}}{\ln \alpha}$.
\label{prop:line.upper}
\end{proposition}

\begin{proof}[Proof of Proposition~\ref{prop:line.upper}]
We say that {\em iteration $j$ is on $P$} if during that iteration, the Searcher searches the half line $P$. The cost at which $s_{a,\mu}$ locates a target that hides at the predicted halfline, say on iteration $i+2$, is at most
\begin{align*}
c(s_{a,\mu},H) &=|H|+2E[\sum_{j=-\infty}^{i+1}x_j \ | \ x_i <|H| \leq x_{i+2}, \ i \ \text{on $P$}] \\
&=|H| +2E[(1+a\mu)\frac{1}{1-\frac{1}{a^2}}x_i  | \ x_i <|H| \leq a^2 x_i] \\ 
&=|H|+2(1+a\mu)\frac{a^2}{a^2-1} E[x_i |\ \ \frac{|H}{a^2} <x_i \leq |H| \\
&\leq |H|+2(1+a\mu)\frac{a^2}{a^2-1} \int_0^2\frac{1}{2}|H| a^{u-2} du \\
&= |H|(1+\frac{1+a\mu}{1+\ln a}),
\end{align*}
which establishes the consistency bound. 

Similarly, the cost at which $s_{a,\mu}$ locates a target that hides at the predicted halfline, say on iteration $i+2$, is at most
\begin{align*}
c(s_\alpha,H) &=|H|+2E[\sum_{j=-\infty}^{i+1}x_j \ | \ x_i <|H| \leq x_{i+2}, \ i \ \text{not on $P$}] \\
&=|H| +2E[\frac{1}{1-\frac{1}{a^2}}x_i  | \ x_i <|H| \leq a^2 x_i] \\ 
&=|H|+2(1+\frac{a}{\mu})\frac{a^2}{a^2-1} E[x_i |\ \ \frac{|H}{a^2} <x_i \leq |H| \\
&\leq |H|+2(1+\frac{a}{\mu})\frac{a^2}{a^2-1} \int_0^2\frac{1}{2}|H| a^{u-2} du \\
&= |H|(1+\frac{1+\frac{a}{\mu}}{1+\ln a}),
\end{align*}
which establishes the robustness bound. 
\end{proof}


We will now show that the strategies $s_{\alpha, \mu}$ of Definition~\ref{def:upper.linear} can characterize the Pareto frontier of the game. Specifically, by Lemma~\ref{lem:pareto.game} , it is sufficient to show that for each $ \lambda \in [0,1]$, the Searcher has an optimal strategy of the form $s_{\alpha, \mu}$, for some $\mu$ that is a function of $\lambda$. Recall that the payoff of the game $G_{(1,\lambda)}$ for any $s \in X$ and $(y_1,y_2) \in H \times Y$ is
\[
v(s,(y_1,y_2)) = u(s,y_1)+\lambda u(s,y_2).
\]
It is easy to see, using Proposition~\ref{prop:line.upper}, that for fixed $\alpha$, the best upper bound on the value of $G_{(1,\lambda)}$ obtained by this family of strategies is by taking $\mu=\sqrt{\lambda}$. Hence the following corollary:
\begin{corollary}
For fixed $\lambda \in [0,1]$, the strategy $s_{\alpha, \sqrt{\lambda}}$ ensures a payoff 
\begin{align}
v(s_{\alpha, \sqrt{\lambda}}, (y_1,y_2)) \le 1+\lambda +  \frac{1+\lambda+ 2 \sqrt \lambda \alpha}{\ln \alpha} \label{eq1}
\end{align}
in $G_{(1,\lambda)}$, for any $y_1 \ge 0$ and $y_2 \neq 0$.
\label{cor:linear.upper}
\end{corollary}
It is easy to see that for any fixed $\lambda$, the RHS of~\eqref{eq1} has a unique minimum attained for some $\alpha=\bar{\alpha}>1$, hence
\begin{equation}
v(s_{\alpha, \sqrt{\lambda}}, (y_1,y_2)) \le 1+\lambda +  \frac{1+\lambda+ 2 \sqrt \lambda\bar{\alpha}}{\ln \bar{\alpha}}.
\label{eq:eq2}
\end{equation}

In the next step, we will derive a lower bound that matches~\eqref{eq:eq2} using Lemma~\ref{lem:pareto.game}. Namely, we will present a 
family of mixed Hider strategies that come arbitrarily close to solving the problem 
\[
\max_{(h_1,h_2)\in \Delta(H \times Y)} \min_{S \in X} (u(s,h_1) + \lambda u(s,h_2)).
\]
To this end, let $S$ denote a pure (deterministic) strategy. 
We will prove the following lower bound (Theorem~\ref{thm:linear.lower}), by extending an approach for the standard game due to Gal~\cite{Gal80}. The intuition is to define two {\em asymmetric} hider distributions so as to reflect that the prediction $H$ points to the right half-line. For the given $\lambda$, the distribution spreads the hider to a shorter span in the left half-line than the right one. Specifically, these spans are multiplied and divided by $\sqrt{\lambda}$, respectively, as will become clear in the proof. 

\begin{theorem}
For any fixed $\varepsilon>0$, and any $\lambda \in (0,1)$, there exist hiding distributions $h_1 \in X$
and $h_2 \in Y$ such that for any pure strategy $S$ it holds that
\[
 u(S,h_1) + \lambda u(S,h_2)   \ge  (1-\varepsilon) \left( 1 +\lambda + \frac{1+\lambda + 2\sqrt \lambda \bar{\alpha}}{ \ln \bar{\alpha}} \right).
\]
\label{thm:linear.lower}
\end{theorem}

\begin{proof}[Proof sketch]
Let $\varepsilon >0$ and let $R=\exp((1-\varepsilon)/\varepsilon)$. (Assume that $\varepsilon$ is sufficiently small so that $R >1$.) For $\alpha > 0$, let $h_{\alpha}$ be the hiding distribution defined on the interval $[\alpha, \alpha R]$ that has p.d.f. $\varepsilon/x$ for $\alpha \le x \le \alpha R$ and an atom of mass $\varepsilon$ at $x=\alpha R$.  It is easy to verify that this is a well defined probability distribution.  Also write $h_{-\alpha}$ for the hiding strategy on the interval $[-\alpha R, -\alpha]$ that has p.d.f. $-\varepsilon/x$ for $-\alpha R \le x \le -\alpha$ and an atom of mass $\varepsilon$ at $x=-\alpha R$

Note that for $\alpha >0$ and $\alpha \le y \le \alpha R$,
\begin{align}
\int_{y}^{\alpha R} \frac{1}{x} ~ dh_{\alpha}(x) = \int_{y}^{\alpha R} \frac{1}{x} \cdot \frac{\varepsilon}{x} ~ dx + \frac{\varepsilon}{\alpha R}
= \frac{\varepsilon}{y}. \label{eq:hiding-int}
\end{align}

Let $S$ be an arbitrary strategy in which the Searcher first goes to $x_0 \ge 1/{\sqrt \lambda}$, then to $-x_1 \le -\sqrt \lambda$, then to $x_2 \ge x_0$, then to $-x_3 \le -x_1$, and so on until reaching $-x_{n-1}=-\sqrt \lambda R$ and $x_n = R/{\sqrt \lambda}$, for some $n$. Note that, for a strategy that is defined in this way, $n$ must be even. We will first prove the result under the assumption that 
the strategy is of this form, then in the Appendix we extend the proof to account for the remaining cases. For convenience, we set $x_{-1}=\sqrt{\lambda}$ and $x_{-2}=1/\sqrt{\lambda}$.

Note that for $0 \le i \le n$ and $x_{i-2} < x \le x_i$, we have
\begin{align}
u(S,x) = \frac{x + 2\sum_{j=0}^{i-1} x_j}{x} = 1 + \frac{2}{x} \sum_{j=0}^{i-1} x_j. \label{eq:cr}
\end{align}

Using~(\ref{eq:cr}), we obtain an expression for $u(s, h_{1/\sqrt{\lambda}})$ as follows.

\begin{align*}
u(S, h_{1/\sqrt{\lambda}}) &= \sum_{i \text{ even, } 0 \le i \le n} \int_{x_{i-2}}^{x_i} \left( 1 + \frac{2}{x} \sum_{j=0}^{i-1} x_j\right) ~ dh_{1/\sqrt{\lambda}}  \\
&=1 + 2 \cdot \sum_{i \text{ even, } 0 \le i \le n-2} x_i \int_{x_{i}}^{x_n} \frac{1}{x} ~ dh_{1/\sqrt{\lambda}}(x)  + 2 \cdot \sum_{i \text{ odd, } 1 \le i \le n-1} x_i  \int_{x_{i-1}}^{x_n} \frac{1}{x} ~ dh_{1/\sqrt{\lambda}}(x), 
\end{align*}
by swapping the order of integration and summation. Using (\ref{eq:hiding-int}) and the definition of $x_n$, this simplifies to
\begin{align}
u(S, h_{1/\sqrt{\lambda}}) & = 1 + 2 \cdot \sum_{i \text{ even, } 0 \le i \le n-2}   \frac{\varepsilon x_i}{x_{i}} + 2 \cdot \sum_{i \text{ odd, } 1 \le i \le n-1} \frac{ \varepsilon x_i}{x_{i-1}} \nonumber \\
& = 1 + \varepsilon n + 2\varepsilon \cdot \sum_{i \text{ odd, } 1 \le i \le n-1} \frac{ x_i}{x_{i-1}} \label{eq:+ve}.
\end{align}

Similarly (and using the symmetry of $h_{-\sqrt{\lambda}}$ and $h_{\sqrt{\lambda}}$),

\begin{align}
u(S, h_{-\sqrt{\lambda}}) &= \sum_{i \text{ odd, } 1 \le i \le n-1} \int_{x_{i-2}}^{x_i} \left( 1 + \frac{2}{x} \sum_{j=0}^{i-1} x_j\right) ~ dh_{\sqrt{\lambda}} \nonumber \\
&=1 + 2 \cdot \sum_{i \text{ even, } 0 \le i \le n-2} x_i \int_{x_{i-1}}^{x_{n-1}} \frac{1}{x} ~ dh_{\sqrt{\lambda}}(x)  + 2 \cdot \sum_{i \text{ odd, } 1 \le i \le n-1} x_i  \int_{x_{i}}^{x_{n-1}} \frac{1}{x} ~ dh_{\sqrt{\lambda}}(x) \nonumber \\
&= 1 + \varepsilon n + 2\varepsilon \cdot \sum_{i \text{ even, } 0 \le i \le n-2} \frac{ x_i}{x_{i-1}} \label{eq:-ve}.
\end{align}

Combining Equations~(\ref{eq:+ve}) and~(\ref{eq:-ve}), 
\begin{align*}
u(s,h_{1/\sqrt{\lambda}}) + \lambda u(S,h_{-\sqrt{\lambda}}) &= 1+\lambda + \varepsilon n (1+\lambda) + 2\varepsilon \cdot \sum_{i \text{ even, } 0 \le i \le n-2} \frac{\lambda x_i}{x_{i-1}} +  2\varepsilon \cdot \sum_{i \text{ odd, } 1 \le i \le n-1} \frac{x_i}{x_{i-1}} \\
& \ge 1+\lambda + \varepsilon n(1+\lambda) + 2 \varepsilon n \left( \lambda^{n/2} \prod_{i=0}^{n-1} \frac{x_i}{x_{i-1}} \right)^{1/n},
\end{align*}
using the inequality of arithmetic and geometric means. The product is telescopic, thus
\begin{align*}
 u(S,h_{1/\sqrt{\lambda}}) + \lambda u(S,h_{-\sqrt{\lambda}})  & = 1 + \lambda + \varepsilon n (1+\lambda) + 2 \varepsilon n \sqrt{\lambda} R^{1/n} \\
 & = 1 +\lambda + (1-\varepsilon) \frac{1+\lambda + 2\sqrt \lambda R^{1/n}}{ \ln R^{1/n}} \\
  & \ge (1-\varepsilon) \left( 1 +\lambda + \frac{1+\lambda + 2\sqrt \lambda \bar{\alpha}}{ \ln \bar{\alpha}} \right),
\end{align*}
which completes the proof. We recall that we assumed that $S$ starts by going to $x_0 >0$  and $n$ is even. There are three remaining cases to consider, all covered in the Appendix, which follow along the same lines, though with some additional technical considerations.
\end{proof}

Combining Theorem~\ref{thm:linear.lower} and Eq.~\eqref{eq:eq2} we obtain the complete characterization of the Pareto frontier:

\begin{theorem}
The Pareto frontier $P$ for the linear search game is given by 
\[
P= \{\left(1+(1+\mu\bar{\alpha})/\ln \bar{\alpha},1+(1+\bar{\alpha}/\mu)/\ln \bar{\alpha} \right): \lambda \in (0,1) \}
\]
where $\bar{\alpha}$ minimizes the function
$
f(\alpha)= 1+\lambda +  \frac{1+\lambda+ 2 \sqrt \lambda \alpha}{\ln \alpha}.
$
\label{thm:main.linear}
\end{theorem}

\section{Conclusion}
\label{sec:conclusions}

We gave the first study of search games in settings in which the Searcher has some prediction on the Hider's position, and the objective is to identify the Pareto frontier of consistency/robustness tradeoffs. Our model does not require any probabilistic assumptions on the correctness of the prediction, and the analysis reveals the power and the limitations of fully randomized strategies. We presented different approaches for applying the minmax theorem, including one based on the concept of scalarization from multiobjective optimization. Our proposed framework is quite general, and can be applied to any of the multitude of search games that have been studied in the literature, but also to any two-person zero-sum game with predictions. 

There are other classes of games rooted in Search Theory which are natural candidates for a study that incorporates predictions. For example, one may consider {\em patrolling games}, e.g.,~\cite{alpern2011patrolling}, in which the Patroller must guard the environment against the invasion of an Intruder. Here, the Patroller may have some ``prediction'' (e.g., information obtained by intelligence) on the candidate intrusion points and/or the times at which the intrusion will take place, but would like also to safeguard against the possibility that the intelligence is malicious disinformation. Another class of candidate problems is {\em rendezvous} games, and in particular the variant in which the players leave {\em markers} so as to speedup the rendevous, see, e.g.~\cite{baston2001rendezvous}. Here, it would be interesting to consider the possibility that the markers are not necessarily reliable, and describe the Pareto frontier of the optimal rendezvous strategies in such settings. Last, one may consider the well-studied class of {\em cops and robbers} games, e.g.,~\cite{frieze2012variations}, in which the Cop player must capture the Robber player that aims to permanently escape by moving within an environment modeled by a graph. Here, the prediction may provide some information on the movements of the Robber within the graph, and the consistency/robustness can be related, for instance, to the cop number if the predictions are reliable or adversarial, respectively.  

\paragraph{Acknowledgements}
This material is based upon work supported by the Air Force Office of Scientific Research under award number FA9550-23-1-0556, and by the project PREDICTIONS, grant ANR-23-CE48-0010 from the French National Research Agency (ANR). This research also benefited from a DAAD academic fellowship at LMU Munich. The authors would also like to acknowledge the Lorentz Center at Leiden University, since some of the results were obtained at the Workshop on Search Games organized by the Lorentz Center.

\bibliographystyle{plain}
\bibliography{targets}

\bigskip

\appendix

\noindent
{\Large\bf Appendix}

\bigskip

\begin{proof}[Details in the proof of Theorem~\ref{thm:linear.lower}]

We will consider below the three remaining cases in the proof. With a slight change of notation, we will show that for every $\eta>0$ it holds that
\[
 u(S,h_1) + \lambda u(S,h_2)   \ge  (1-\eta) \left( 1 +\lambda + \frac{1+\lambda + 2\sqrt \lambda \bar{\alpha}}{ \ln \bar{\alpha}} \right),
\]

Since some of the remaining cases are technically more complicated, we need to introduce some additional definitions.
Let $\delta=\eta/2$ and let $N$ be an integer such that $\lambda^{1/(2N)} \ge 1-\delta$. Let $R \ge \max\{e, \exp((1-\delta)/\delta)\}$ be large enough to satisfy
\begin{align}
 \frac{1+\lambda + 2\lambda R^{1/N}}{ \ln R^{1/N}} \ge \frac{1+\lambda + 2\sqrt \lambda \bar{\alpha}}{\ln \bar{\alpha}}. \label{eq:Rbig}
\end{align}
Note that this is possible because the expression on the left of (\ref{eq:Rbig}) is increasing in $R$ for $R$ sufficiently large. We need to make sure $R \ge e$ to ensure that $R^{1/n} / \ln R^{1/n}$ is decreasing in $n$ for all $n \ge 1$ (which will be important later). Let  $\varepsilon = 1/(\ln R + 1)$. Then  $R=\exp((1-\varepsilon)/\varepsilon)$ and $\varepsilon < \delta$, since $R \ge \exp((1-\delta)/\delta)$. Observe that with these definitions in mind,~\eqref{eq:hiding-int} still holds. We can now proceed with the analysis of the remaining three cases.

\subsubsection*{Case 2: $n$ is even, $S$ starts by going to $-x_0 < 0$.}

We write $x_{-1}=1/\sqrt{\lambda}$ and $x_{-2}=\sqrt{\lambda}$. 
Similarly to the main case, we derive

\begin{align*}
u(S, h_{-\sqrt{\lambda}}) &= 1 + 2  \sum_{\substack{i \text{ odd,}\\ 1 \le i \le n-1}} (x_{i} + x_{i-1}) \int_{x_{i-1}}^{x_n} \frac{1}{x} ~ dh_{\sqrt{\lambda}}(x) \\
& = 1 + \varepsilon n + 2\varepsilon  \sum_{\substack{i \text{ odd,}\\ 1 \le i \le n-1}} \frac{ x_{i}}{x_{i-1}}.
\end{align*}

Also,

\begin{align*}
 u(S, h_{1/\sqrt{\lambda}}) &= 1 + 2  x_0 \int_{x_{-1}}^{x_{n-1}} \frac{1}{x} ~ dh_{1/\sqrt{\lambda}}(x) + 2 \sum_{\substack{i \text{ even,}\\ 2 \le i \le n-2}} (x_i + x_{i-1}) \int_{x_{i-1}}^{x_{n-1}} \frac{1}{x} ~ dh_{1/\sqrt{\lambda}}(x) \\
& = 1 + \varepsilon (n-2) + 2\varepsilon  \sum_{\substack{i \text{ even,}\\ 0 \le i \le n-2}} \frac{ x_{i}}{x_{i-1}}.
\end{align*}

Hence,

\begin{align*}
u(S,h_{1/\sqrt{\lambda}}) + \lambda u(S,h_{-\sqrt{\lambda}}) &= (1+\lambda)(1 + \varepsilon n) - 2 \varepsilon +   2\varepsilon  \sum_{\substack{i \text{ even,}\\ 0 \le i \le n-2}} \frac{ x_{i}}{x_{i-1}} +  2\varepsilon \lambda  \sum_{\substack{i \text{ odd,}\\ 1 \le i \le n-1}} \frac{ x_{i}}{x_{i-1}} \\
&\ge (1+\lambda)(1 + \varepsilon n - 2 \varepsilon)  + 2 \varepsilon n \left( \lambda^{n/2} \prod_{i=0}^{n-1} \frac{x_i}{x_{i-1}} \right)^{1/n} \\
 & = (1 +\lambda)(1- 2 \varepsilon)  + (1-\varepsilon) \frac{1+\lambda + 2\sqrt \lambda R^{1/n}}{ \ln R^{1/n}} \\
  & \ge (1-\eta) \left( 1 +\lambda + \frac{1+\lambda + 2\sqrt \lambda \bar{\alpha}}{ \ln \bar{\alpha}} \right),
\end{align*}
since $2 \varepsilon < 2\delta= \eta$.

\subsubsection*{Case 3: $n$ is odd, $S$ starts by going to $x_0 > 0$.}

We write $x_{-1}=\sqrt{\lambda}$ and $x_{-2}=1/\sqrt{\lambda}$. We have

\begin{align*}
u(S, h_{1/\sqrt{\lambda}}) &= 1 + 2  \sum_{\substack{i \text{ odd,}\\ 1 \le i \le n-2}} (x_{i} + x_{i-1}) \int_{x_{i-1}}^{x_{n-1}} \frac{1}{x} ~ dh_{1/\sqrt{\lambda}}(x) \\
&= 1+ \varepsilon(n-1) + 2\varepsilon  \sum_{\substack{i \text{ odd,}\\ 1 \le i \le n-2}} \frac{ x_{i}}{x_{i-1}}.
\end{align*}

Also,

\begin{align*}
u(S, h_{-\sqrt{\lambda}}) &= 1 + 2 x_0  \int_{x_{-1}}^{x_{n}} \frac{1}{x} ~ dh_{\sqrt{\lambda}}(x) + 2  \sum_{\substack{i \text{ even,}\\ 2 \le i \le n-1}} (x_{i} + x_{i-1}) \int_{x_{i-1}}^{x_{n}} \frac{1}{x} ~ dh_{\sqrt{\lambda}}(x) \\
&=1 + \varepsilon (n-1) + 2\varepsilon  \sum_{\substack{i \text{ even,}\\ 0 \le i \le n-1}} \frac{ x_{i}}{x_{i-1}}.
\end{align*}

Hence,

\begin{align*}
u(S,h_{1/\sqrt{\lambda}}) + \lambda u(S,h_{-\sqrt{\lambda}}) &= (1+\lambda)(1 + \varepsilon n - \varepsilon)  +   2\varepsilon  \sum_{\substack{i \text{ even,}\\ 0 \le i \le n-1}} \frac{ x_{i}}{x_{i-1}} +  2\varepsilon \lambda  \sum_{\substack{i \text{ odd,}\\ 1 \le i \le n-2}} \frac{ x_{i}}{x_{i-1}} \\
&\ge (1+\lambda)(1 + \varepsilon n - \varepsilon)   + 2 \varepsilon n \left( \lambda^{(n-1)/2} \prod_{i=0}^{n-1} \frac{x_i}{x_{i-1}} \right)^{1/n} \\
& = (1+\lambda)(1 + \varepsilon n - \varepsilon) +   2 \varepsilon n \sqrt{\lambda} \cdot \lambda^{-3/(2n)} R^{1/n} \\
 & \ge (1 +\lambda)(1- \varepsilon) + (1-\varepsilon) \frac{1+\lambda + 2\sqrt \lambda R^{1/n}}{ \ln R^{1/n}} \\
  & \ge (1-\eta) \left( 1 +\lambda + \frac{1+\lambda + 2\sqrt \lambda \bar{\alpha}}{ \ln \bar{\alpha}} \right) .
\end{align*}

\subsubsection*{Case 4: $n$ is odd, $S$ starts by going to $-x_0 < 0$.}

We write $x_{-1}=1/\sqrt{\lambda}$ and $x_{-2}=\sqrt{\lambda}$. Again, leaving out some of the details,

\begin{align*}
u(S, h_{-\sqrt{\lambda}}) &= 1 + 2  \sum_{\substack{i \text{ odd,}\\ 1 \le i \le n-2}} (x_{i} + x_{i-1}) \int_{x_{i-1}}^{x_{n-1}} \frac{1}{x} ~ dh_{\sqrt{\lambda}}(x) \\
& = 1 + \varepsilon (n-1) + 2\varepsilon  \sum_{\substack{i \text{ odd,}\\ 1 \le i \le n-2}} \frac{ x_{i}}{x_{i-1}}.
\end{align*}

Also,

\begin{align*}
 u(S, h_{1/\sqrt{\lambda}}) &= 1 + 2  x_0 \int_{x_{-1}}^{x_{n}} \frac{1}{x} ~ dh_{1/\sqrt{\lambda}}(x) + 2 \sum_{\substack{i \text{ even,}\\ 2 \le i \le n-1}} (x_i + x_{i-1}) \int_{x_{i-1}}^{x_{n}} \frac{1}{x} ~ dh_{1/\sqrt{\lambda}}(x), \\
& = 1 + \varepsilon (n-1) + 2\varepsilon  \sum_{\substack{i \text{ even,}\\ 0 \le i \le n-1}} \frac{ x_{i}}{x_{i-1}}.
\end{align*}

Hence,

\begin{align}
u(S,h_{1/\sqrt{\lambda}}) + \lambda u(S,h_{-\sqrt{\lambda}}) &= (1+\lambda)(1 + \varepsilon n - \varepsilon)  +   2\varepsilon  \sum_{\substack{i \text{ even,}\\ 0 \le i \le n-1}} \frac{ x_{i}}{x_{i-1}} +  2\varepsilon \lambda  \sum_{\substack{i \text{ odd,}\\ 1 \le i \le n-2}} \frac{ x_{i}}{x_{i-1}} \nonumber \\
&\ge (1+\lambda)(1 + \varepsilon n - \varepsilon)   + 2 \varepsilon n \left( \lambda^{(n-1)/2} \prod_{i=0}^{n-1} \frac{x_i}{x_{i-1}} \right)^{1/n} \nonumber \\
& = (1+\lambda)(1 + \varepsilon n - \varepsilon) +   2 \varepsilon n \sqrt{\lambda} \cdot \lambda^{1/(2n)} R^{1/n} \nonumber \\
 &  = (1 +\lambda)(1- \varepsilon) + (1-\varepsilon) \frac{1+\lambda + 2\sqrt \lambda \cdot \lambda^{1/(2n)} R^{1/n}}{ \ln R^{1/n}}. \label{eq:twocases}
\end{align}
We now split the analysis into two further cases. The first case is if $n \ge N$, so that $\lambda^{1/(2n)} \ge \lambda^{1/(2N)} \ge 1-\delta$, by definition of $N$. Then (\ref{eq:twocases}) implies that
\begin{align*}
u(S,h_{1/\sqrt{\lambda}}) + \lambda u(S,h_{-\sqrt{\lambda}}) & \ge (1 +\lambda)(1- \varepsilon) + (1-\varepsilon) \frac{1+\lambda + 2\sqrt \lambda \cdot (1-\delta) R^{1/n}}{ \ln R^{1/n}} \\
&\ge (1-\delta)^2  \left( 1 +\lambda + \frac{1+\lambda + 2\sqrt \lambda R^{1/n}}{ \ln R^{1/n}} \right) \\
& \ge (1 - \eta) \left( 1 +\lambda + \frac{1+\lambda + 2\sqrt \lambda \bar{\alpha}}{ \ln \bar{\alpha}} \right),
\end{align*}
since $(1-\delta)^2 = 1-\eta + \delta^2 \ge 1-\eta$.

The second case is if $n <N$. In this case, from (\ref{eq:twocases}) and $(\ref{eq:Rbig})$, and using $\lambda^{1/(2n)} \ge \sqrt{\lambda}$ and the fact that $R^{1/n}/\ln R^{1/n}$ is decreasing in $n$, we have
\begin{align*}
u(S,h_{1/\sqrt{\lambda}}) + \lambda u(S,h_{-\sqrt{\lambda}}) & \ge (1 +\lambda)(1- \varepsilon) + (1-\varepsilon) \frac{1+\lambda + 2 \lambda R^{1/N}}{ \ln R^{1/N}} \\
 & \ge (1 - \varepsilon) \left( 1 +\lambda + \frac{1+\lambda + 2\sqrt \lambda \bar{\alpha}}{ \ln \bar{\alpha}} \right) \\
 & \ge (1 - \eta) \left( 1 +\lambda + \frac{1+\lambda + 2\sqrt \lambda \bar{\alpha}}{ \ln \bar{\alpha}} \right) .
\end{align*}

We conclude that the result holds in all possible cases.
\end{proof}

\end{document}